\documentclass[11pt]{article}
\usepackage{fullpage}
\usepackage[utf8]{inputenc} 
\usepackage[T1]{fontenc}
\usepackage{url}
\usepackage{ifthen}
\usepackage{amsthm}
\usepackage{amssymb}
\usepackage{thmtools} 
\usepackage{thm-restate}
\usepackage[cmex10]{amsmath} 
\usepackage[noend]{algpseudocode}
\usepackage{etoolbox}
\usepackage{appendix}
\usepackage{float}
\usepackage{graphicx}
\usepackage{amsfonts}
\usepackage{xcolor}
\usepackage{balance}
\usepackage{biblatex}
\usepackage{lipsum}
\usepackage[colorlinks=true, linkcolor=black, citecolor=black, urlcolor=cyan]{hyperref}
\usepackage[nameinlink]{cleveref}
\AtBeginEnvironment{appendices}{\crefalias{section}{appendix}}

\makeatletter
\renewcommand{\footnoterule}{%
  \kern -3pt
  \hrule width \columnwidth
  \kern 2.6pt
}
\makeatother

\usepackage[inline,shortlabels]{enumitem}
\usepackage[ruled, vlined, linesnumbered]{algorithm2e}
\SetKwInOut{Input}{Input}
\SetCommentSty{textrm}

\bibliography{biblatex}

\newtheorem{proposition}{Proposition}
\newtheorem{lemma}{Lemma}

\newtheorem{claim}{Claim}

\theoremstyle{definition}
\newtheorem{remark}{Remark}

\begin{document}

\onecolumn

\title{\LARGE  Testing the Independent Set Property in Hypergraphs}


\author{}

\author{
    Elena Grigorescu \\
    Cheriton School of Computer Science\\ 
    University of Waterloo \\
    Email: elena-g@uwaterloo.ca
    \and
    Shreya Nasa \thanks{Supported in part by NSF Awards CCF-2127806, CCF-2228814, and ONR Award N00014-24-1-2695} \\
    Department of Computer Science \\
    Purdue University \\
    Email: snasa@purdue.edu
    \and Cameron Seth \\
    Cheriton School of Computer Science\\ 
    University of Waterloo \\
    Email: cjmpseth@uwaterloo.ca
    \and
}

\date{}

\maketitle

\begin{abstract}
    The optimal sample complexity of testing if an $n$-vertex graph has an independent set of size $\rho n$, or is $\varepsilon$-far from having an independent set of size $\rho n$, was established to be $\widetilde{O}(\rho^3/\varepsilon^2)$, in a notable result by Blais and Seth (SICOMP 2025). In contrast, for $q$-uniform hypergraphs, there is a significant gap between the best known upper and lower bounds, and there has been no progress on the problem for the last two decades. In this work, we prove a new upper bound of $\widetilde{O}\!\left(\frac{q\rho^{2q-3}}{\varepsilon^2 (q-2)!^2}\right)$ on the sample complexity of testing the $\rho$-independent set property. The previous best known upper bound was $\widetilde{O}\!\left(\frac{2^q q! \rho^{2q}}{\varepsilon^3}\right)$, due to Langberg (RANDOM 2004). This establishes the optimal dependence on $\varepsilon$ and gives an exponential improvement in the dependence on $q$.
    We prove our result via a new application of the hypergraph container method.
\end{abstract}

\newpage

\section{Introduction}

Graph and hypergraph property testing address the question of when one can distinguish between the case that a (hyper)graph has some property and the case that it is \emph{far} from having the property, by inspecting only a small fraction of the input (hyper)graph.
In this work we study the classic problem of testing whether a hypergraph has a large independent set in the dense hypergraph model.
We say that a $q$-uniform hypergraph $H$ on $n$ vertices has the $\rho$-independent set property if $H$ has an independent set $I$ of size $\rho n,$ and $H$ is considered to be $\varepsilon$-far from having a $\rho n$ independent set if at least $\varepsilon n^q$ edges must be removed from $H$ to obtain a hypergraph with an independent set of size $\rho n$.

An $\varepsilon$-\emph{tester} for $\rho$-independence is a bounded error randomized algorithm that takes a random sample $S \subseteq V$ of vertices, inspects $H[S],$ and, if $H$ has a $\rho n$ independent set, accepts with probability at least $\delta,$ and if $H$ is $\varepsilon$-far from having a $\rho n$ independent set, accepts with probability at most $\gamma,$ for some positive constants $\delta > \gamma$.\footnote{The standard success acceptance probabilities are $\delta=2/3$ and $\gamma=1/3,$ but so long as $\delta > \gamma$ the probabilities can be boosted to $2/3$ and $1/3$ by standard error reduction techniques. We note that if $\delta=1$ the tester is called a one-sided error tester, and otherwise it's called a two-sided error tester.}
The sample size $|S|$ is known as the sample complexity of the tester, and the sample complexity of testing the $\rho n$ independent set property is the minimum sample complexity of any such $\varepsilon$-tester for the property.

Independent set testing on graphs (when $q=2$) is one of the classic graph properties which was studied by Goldreich, Goldwasser and Ron when they introduced the notion of graph property testing \cite{goldreichPropertyTesting1998}.
They showed that it is possible to $\varepsilon$-test the $\rho$-independent set property in graphs with sample complexity $\widetilde{O}(\rho/\varepsilon^4).$
Feige, Langberg and Schechtman later improved the bound to $\widetilde{O}(\rho^4/\varepsilon^3)$ \cite{feigeCliqueTesting2004}.
The problem of finding the optimal sample complexity in graphs remained open for nearly 20 years until a recent work by Blais and Seth improved the bound to $\widetilde{O}(\rho^3/\varepsilon^2)$ by introducing a new connection between property testing and the graph container method \cite{blais2025testing}.
The result by Blais and Seth is optimal up to logarithmic (in $1/\varepsilon$) factors based on a lower bound in \cite{feigeCliqueTesting2004}.

Finding the optimal sample complexity of $\rho$-independent set testing in $q$-uniform hypergraphs remains open.
Generalizing the results of \cite{feigeCliqueTesting2004}, Langberg showed that there is an $\varepsilon$-tester for $\rho$-independent set testing in hypergraphs with sample complexity of $\widetilde{O}\left(\frac{2^q q! \rho^{2q}}{\varepsilon^3}\right)$ \cite{Langberg04}.
Additionally, Langberg gave a lower bound on the sample complexity of $\widetilde{\Omega}\left(\frac{\rho^{2q-1}}{(q!)^2 \varepsilon^2}\right).$
This leaves the natural open question: what is the optimal dependence on $\varepsilon$ for the sample complexity of testing $\rho$-independent sets in $q$-uniform hypergraphs?
Further, there is also an exponential gap in the $q$ dependence between the upper and lower bounds.
It's worth noting that even though the upper bound has a $\rho^{2q}$ in the numerator and the lower bound has a $\rho^{2q-1}$ in the numerator, the problem is only non-trivial when $\varepsilon < \frac{\rho^q}{q!}$ and so the lower bound is strictly smaller than the upper bound.

\subsection{Our Results}
We give a new upper bound on the sample complexity of testing the $\rho $-independent set property in hypergraphs.

\begin{restatable}{theorem}{upperBound}
\label{thm1}
    The sample complexity for testing the $\rho$-independent set property is $\widetilde{O}\left(\frac{q \rho^{2q-3}}{\varepsilon^2 (q-2)!^2}\right)$.
\end{restatable}

Our analysis achieves the optimal dependence on $\varepsilon$, matching the lower bound of Langberg, while also improving the dependence on $q$ exponentially compared to previous results.

There is a gap between our upper bound and the lower bound of Langberg of roughly $\frac{\rho^2}{q^5}.$
We note that the lower bound result of Langberg can be improved slightly to $\frac{\rho^{2q-1}}{((q-1)!)^2 \varepsilon^2}$ using the same proof approach (see \Cref{appendix1}), and so this reduces the gap between upper and lower bounds to roughly $\frac{\rho^2}{q^3}.$
Achieving the optimal $\rho$ dependence remains an open problem, however we show in \Cref{remark:containerLemmaTightness} that our analysis cannot be tightened to achieve the optimal $\rho$-dependence.

We note that, while our bound applies for any $q,$ in the regime that $q$ is non-constant, there is an alternate distance measure.
We could instead say a $q$-uniform hypergraph is $\gamma$-far from having a $\rho n$ independent set if at least $\gamma \binom{n}{q}$ edges must be modified to make the hypergraph have the property, similar to what was studied by \cite{blais2024new} for satisfiability testing.
Observe that $\gamma \binom{n}{q} \approx \frac{\gamma}{q!}n^q,$ and so if a hypergraph is $\gamma$-far from having a $\rho n$ independent set under this alternate distance measure then it is $\varepsilon \approx \frac{\gamma}{q!}$-far in the distance measure we use in this paper.
Using \Cref{thm1} with this change of variable we get an upper bound of roughly $\frac{q^5 \rho^{2q-3}}{\gamma^2}$ for testing $\rho$-independence under this alternate notion of distance.

\subsection{Overview of Techniques}
To prove \Cref{thm1} we analyze the natural testing algorithm that takes a random sample $S \subseteq V$ of vertices and accepts if and only if $H[S]$ contains an independent set of size $\rho |S|.$
This is a two-sided error tester, which is necessary for any tester for the $\rho$-independent property, as it is not a semihereditary property \cite{alon2008characterization}.
As is common for property testing algorithms, the main challenge is showing that the tester rejects hypergraphs that are $\varepsilon$-far from having a $\rho n$ independent set with high probability.

The main tool in our analysis is a new application of the hypergraph container method.
Blais and Seth first applied the graph container method to analyze property testers for the $\rho$-independent set and $k$-colorability properties on graphs \cite{blais2025testing}. They further used the hypergraph container method to analyze the sample complexity for the satisfiability problem and hypergraph coloring \cite{blais2024new}. We adapt their hypergraph container method algorithm to analyze the sample complexity of $\rho$-independence testing in hypergraphs. We first briefly describe the container method and the approach used by Blais and Seth for graphs, and then discuss the main challenges with adapting the container method to analyzing testers for the $\rho$-independent set property on hypergraphs.

\subsubsection{The Container Method}
The container method is a tool for characterizing independent sets in graphs with specific structure, commonly regular or nearly regular graphs.
While in general a graph on $n$ vertices may have up to $2^n$ independent sets, the container method can be used to show that for a graph $G$ satisfying some specific structure, there is a relatively small collection of containers, where each container is a subset of vertices and each container is relatively small, such that every independent set is a subset of a container.

The natural way to construct the collection of containers is with a container generating algorithm, first used by Kleitman and Winston \cite{KleitmanWinston1982}, which works as follows.
Given a graph $G=(V,E)$ and an independent set $I\subseteq V,$ the goal of the container generating algorithm is to find a small fingerprint $F \subseteq I$ that can be used to construct a small container $C \subseteq V$ such that $I \subseteq C.$
In graphs, the natural approach is an iterative greedy approach that first initializes the candidate container to $V,$ and in each step adds a vertex $v \in I$ to the fingerprint which has highest degree in the current container candidate set.
Then, all the neighbours of $v$ are removed from the current container because they cannot be in the independent set.
Further, any vertex with degree higher than $v$ in the current container can be removed as well because they also cannot be in the independent set (otherwise they would have been selected to the fingerprint).
The container generating procedure stops once the container is sufficiently small.
This procedure is run on every independent set to construct the collection, however, since each constructed fingerprint corresponds to a unique container, and each fingerprint is relatively small, then there are a small number of containers (corresponding to the total number of possible fingerprints).

\subsubsection{Overview of Blais--Seth Approach}
Blais and Seth showed that the container method applies very naturally to graphs that are $\varepsilon$-far from having a $\rho n$ independent set \cite{blais2025testing}.
They showed that when $G$ is $\varepsilon$-far from having a $\rho n$ independent set, there is a collection of containers, where each container is a subset of vertices, such that every independent set is a subset of a container.
They show that each container is small (strictly smaller than $\rho n$) and that there are not too many of them (at most roughly $n^{1/\varepsilon}$).

Blais and Seth then use this collection of containers to analyze the natural testing algorithm that takes a random sample $S \subseteq V$ of vertices and checks if $G[S]$ has an independent set of size $\rho |S|.$
In particular, when $G$ is $\varepsilon$-far from having a $\rho n$ independent set, they use the collection of containers to upper bound the probability that $G[S]$ has a $\rho |S|$ independent set as follows.
By the container statement, since every independent set is a subset of a container, $G[S]$ only has a $\rho |S|$ independent set if at least $\rho |S|$ vertices of $S$ are drawn from a specific container.
The probability that $\rho |S|$ vertices are drawn from a specific container can be upper bounded by a Chernoff bound, and then using a union bound they upper bound the probability that $\rho |S|$ vertices are drawn from any container, thereby upper bounding the probability that $G[S]$ has a $\rho |S|$ independent set.

In order to get the optimal bound for $\rho n$ independent set testing, Blais and Seth are crucially able to show that there is a trade off in the container and fingerprint size.
In particular, they show that large containers have small fingerprints and vice versa.
A key step in their proof is an argument that roughly says the following: given a graph $G=(V,E)$ that is $\varepsilon$-far from having an independent set and some independent set $I$, if the container generating procedure ends up at a very large (and sparse) container $C \subseteq V,$ then either $|E(C,V \setminus C)|$ is relatively large, or $G[V \setminus C]$ is very dense.
They show that in either case a large number of vertices would be removed from the container in a single step, and, further, this holds at every step of the container procedure.
Since many vertices are removed in each step of the procedure then the procedure must terminate in a relatively small number of steps, resulting in a relatively small fingerprint.

\subsubsection{First Challenge with Adapting to the Hypergraph Problem}
One immediate challenge for hypergraphs is that the container generating procedure is not as simple as in the graph case.
In graphs, selecting a fingerprint vertex to be the vertex in $I$ with highest degree in the current container works very well because if $v \in I$ then the neighbours of $v$ must not in the independent set and so can immediately be removed from the current container.
However, even for $3$-uniform hypergraphs this does not work: if $v \in I$ then vertices that share an edge with $v$ may still be in the independent set.

We address this challenge using a similar approach as used by Blais and Seth in their follow-up work on testing satisfiability \cite{blais2024new}, and also in hypergraph container method literature (see \cite{baloghIndependentSetsHypergraphs2015,saxtonThomasonHypergraphContainers2015} for example).
In particular, in each iteration of our container generating procedure, we have an inner loop which takes $q$-edges and turns them into $(q-1)$-edges, and repeats on the corresponding $(q-1)$-uniform hypergraph.
For example, in the $3$-uniform hypergraph case, if $v \in I$ then for any pair $\{u,w\}$ for which $\{v,u,w\}$ is an edge in $H,$ at most one of $u$ or $w$ can be in $I.$
So, after selecting a vertex $v \in I$ to the fingerprint, the container procedure will construct a graph (i.e.\ a $2$-uniform hypergraph) where the edges are all pairs $\{u,w\}$ such that $\{v,u,w\}$ is an edge in $H$.
The procedure then can operate on the graph as before by selecting another vertex in $I$ and removing neighbours (see \Cref{sect:NewContainerLemma} for more details).

\subsubsection{Second Challenge with Adapting to the Hypergraph Problem}
The second key challenge in adapting the Blais--Seth approach to get nearly tight bounds for hypergraph $\rho$-independent set testing is that the argument about the trade off between the container and fingerprint size cannot be generalized directly to hypergraphs.
As mentioned above, they show that given a graph $G$ that is $\varepsilon$-far from having an independent set and some independent set $I$, if the container procedure ends up at a very large (and sparse) container $C,$ then either $E(C,V \setminus C)$ is relatively large, or $G[V \setminus C]$ is very dense.
This step of the proof does not generalize naturally to hypergraphs because the edges between $C$ and $V \setminus C$ can now take many forms (i.e. $i$ endpoints in $C$ and $(q-i)$ endpoints in $V \setminus C$ for any $1 \leq i \leq q-1$).
In fact, using a standard hypergraph container generating procedure to analyze a tester for $\rho$-independent set would result in a sample complexity that is larger than our upper bound by at least roughly $1/\rho^q$ (see \Cref{remark:oldApproachWeakness} for an example to demonstrate this).

In order to get our tighter upper bound on the sample complexity we change the hypergraph container procedure following ideas from \cite{blais2024new}.
Instead of selecting a vertex with highest degree in the hypergraph to the fingerprint, we select a vertex to the fingerprint with highest degree in some specific $\rho n$-subhypergraph.
Then the inner loop focuses on this specific $\rho n$-subhypergraph.  We note that Blais and Seth \cite{blais2024new} use a similar approach in the hypergraph container method by focusing on a smaller $n$-vertex subhypergraph in a $kn$-vertex graph for their CSP hypergraph construction.

We note that once we make this change to the container procedure, the proof approach is similar to that used by Blais and Seth to study $k$-colorability and satisfiability \cite{blais2024new}.
Roughly speaking the proof goes by contradiction.
Given a hypergraph $H$ that is $\varepsilon$-far from having a large independent set, and an independent set $I,$ we assume that for all steps of the procedure the container is larger than desired.
We prove a proposition which shows that the max degree in the container decreases in each step, and use this to reconstruct a $\rho n$-subhypergraph that demonstrates that $H$ is not $\varepsilon$-far from having a $\rho n$ independent set, a contradiction. Since we reconstruct a subhypergraph of size $\rho n$ to prove the contrapositive, we obtain a tighter bound by considering the maximum $q$-degree of a vertex in any $\rho n$-induced subhypergraph in the container.

\section{New Container Lemma}
\label{sect:NewContainerLemma}

In this section, we prove a new container lemma for hypergraphs, which will be used in \Cref{sect:thmProof} to prove our upper bound $\widetilde{O}\left(\frac{q \rho^{2q-3}}{\varepsilon^2 (q-2)!^2}\right)$ on the sample complexity. We start by describing our hypergraph container procedure and formalizing some basic properties of the fingerprints and the containers which is then used to prove the new container lemma for $q$-uniform hypergraphs.
We note that Blais and Seth \cite{blais2024new} use similar ideas to prove their upper bound on the sample complexity of CSP (Constraint Satisfaction Problem).

We defer the proof for the lower bound $\Omega\left(\frac{\rho^{2q-1}}{\varepsilon^{2} (q-1)!^{2}}\right)$ using the construction by Langberg \cite{Langberg04} to \Cref{appendix1}.

\subsection{The Container Method}
\label{sect:ContainerMethod}

Our hypergraph container algorithm builds on the work of Kleitman and Winston \cite{KleitmanWinston1982} and of Sapozhenko \cite{Sapozhenko2005}. The key idea behind the container method is to identify a small set of vertices called the \textit{fingerprint} vertices which uniquely identify its \textit{containers} (upto a predefined ordering of vertices) for an input independent set $I$ in a $q$-uniform hypergraph $H = (V, E)$. We denote the fingerprint and container at $t^{th}$ iteration of the algorithm as $F_t(I)$ and $C_t(F_t(I))$ respectively. We may use the shorthand notation $F_t$ and $C_t(F_t)$ when the underlying independent set is clear from the context.

The algorithm (see Algorithm \ref{alg:main-alg}) starts with initializing the entire vertex set as the container $C_0$ and an empty set of vertices as the fingerprint $F_0$. We add vertices to the fingerprint and remove vertices from its corresponding container as the algorithm progresses. At $t^{th}$ iteration, the algorithm chooses the vertex $v_{t}$ from the container $C_{t-1}$ as the vertex in the input independent set $I \subseteq V$ with the maximum $q$-degree in any $\rho n$-induced subhypergraph of the container $C_{t-1}$ and $\mathcal{D}_{t,q}$ as the corresponding subhypergraph. We break ties by giving higher precedence to lower ranked vertices in the ordering $\pi$. We save all the vertices which have higher $q$-degree than the chosen vertex $v_t$ in any $\rho n$ induced subhypergraph induced by vertices in the container $C_{t-1}$ to the set $X_q$. We then focus on this $\rho n$ induced hypergraph to choose the next $(q-1)$ fingerprint vertices. We initialize the fingerprint for the current iteration $F^*$ as $F_{t-1}$ at the start of the inner loop.

For each $\ell^{th}$ (as $\ell$ goes from $q-1$ to 1) iteration of the inner loop, note that $\mathcal{D}_{t, \ell + 1}$ is an $(\ell+1)$-uniform hypergraph with at most $\rho n - (q - (\ell+1))$ vertices. We choose the next fingerprint vertex $v_{t,\ell}$ from $I \setminus F^*$ as the vertex with the maximum $(\ell + 1)$-degree in $\mathcal{D}_{t, \ell + 1}$ and add it to $F^*$. Note that we define $\deg_{\mathcal{D}_{t,\ell+1}}(v)$ for a vertex $v \notin \mathcal{D}_{t,\ell+1}$ as zero. We then construct $\mathcal{D}_{t, \ell}$ by removing $v_{t,\ell}$ from the vertex set of $\mathcal{D}_{t, \ell + 1}$ and using the edges in $\mathcal{D}_{t, \ell + 1}$ incident to $v_{t,\ell}$ minus the vertex $v_{t,\ell}$. We also save the vertices with higher $(\ell + 1)$-degree than the current fingerprint vertex $v_{t,\ell}$ in the current hypergraph $\mathcal{D}_{t,\ell+1}$ to the set $X_\ell$ to be removed later. Note that these vertices cannot be a part of the independent set otherwise they would have been chosen in the fingerprint before the current vertex. We then construct the container set for the next iteration $C_t$ by removing all sets $X_{\ell}$, all vertices with a 1-edge $V[\mathcal{D}_{t,1}]$, and all the new fingerprint vertices $v_{t,\ell}$ from $C_{t-1}$.

\begin{algorithm}[!t]
    \caption{Fingerprint \& Container Generator}
    \label{alg:main-alg}
    \KwIn {A $q$-uniform hypergraph $H = (V, E)$, a parameter $\rho < 1$, an ordering of vertices $\pi$, and an independent set $I \subseteq V$}
    Initialize $F_0 \leftarrow \emptyset$ and $C_0 \leftarrow V(H)$\;
    \For{$t = 1, 2, \dots, |I|/(q - 1)$}{
        \tcp{Choose $v_{t}$ as the vertex with maximum $q$-degree in any $\rho n$ induced subhypergraph of $C_{t-1}$ and $\mathcal{D}_{t,q}$ as the corresponding $\rho n$ induced subhypergraph. We break ties by giving higher precedence to lower ranked vertices in the ordering $\pi$.}
        $(v_{t}, \mathcal{D}_{t,q}) = \underset{{v \in I \cap C_{t-1}, D \subseteq C_{t-1},\, |D| \le \rho n}}{\arg\max} \deg_D(v)$ \;
        \tcp{Keep track of all vertices with higher $q$-degree than $v_{t}$ in any $\rho n$ induced subhypergraph of $C_{t-1}$ to remove later}
        $X_{q} \leftarrow \{ w \in V[C_t] : \deg_{C_{t-1}}^{\leq \rho n}(w) > \deg_{C_{t-1}}^{\leq \rho n}(v_{t}) \}$\;
        \tcp{Initialize the fingerprint for the current iteration}
        $F^* \leftarrow F_{t-1}$ \\
        \For{$\ell = q-1, \dots, 1$}{    
            $v_{t,\ell} \leftarrow$ the vertex in $I \setminus F^*$ with largest $(\ell + 1)$ degree in $\mathcal{D}_{t,\ell+1}$; \\
            $F^* \leftarrow F^* \cup v_{t,\ell}$ \\
            \tcp{Keep track of all vertices with higher $(\ell+1)$-degree than $v_{t,\ell}$ in $\mathcal{D}_{t,\ell+1}$ to remove later}
            $X_{\ell} \leftarrow \{ w \in V[\mathcal{D}_{t,\ell+1}] : \deg_{\mathcal{D}_{t,\ell+1}}(w) > \deg_{\mathcal{D}_{t,\ell+1}}(v_{t,\ell}) \}$\;        
            \tcp{Construct $D_{t,\ell}$ using vertices of $D_{t,\ell+1}$ minus $v_{t, \ell}$ and edges incident on $v_{t, \ell}$}
            $V[\mathcal{D}_{t,\ell}] \leftarrow V[\mathcal{D}_{t,\ell+1}] \setminus \{v_{t,\ell}\}$\;
            $E[\mathcal{D}_{t,\ell}] \leftarrow \{ e \setminus \{v_{t,\ell}\} : e \in E[\mathcal{D}_{t,\ell+1}] \text{ and } v_{t,\ell} \in e \}$\;
        }
        \tcp{Add $v_{t,q-1}, \dots, v_{t,{1}}$ to the fingerprint}
        $F_t \leftarrow F^*$\;
        \tcp{Update container set by removing all sets $X_{\ell}$, all vertices with a 1-edge, and the new fingerprint vertices}
        $C_t \leftarrow C_{t-1} \setminus \left( \bigcup_{\ell=1}^{q} X_\ell \cup \{v \in C_{t-1} : \{v\} \in E[\mathcal{D}_{t,1}]\} \cup \{v_{t,1}, \dots, v_{t,q-1}\} \right)$\;
    }
   \KwRet {$F_1, \dots, F_{|I|/(q-1)} \text{ and } C_1, \dots, C_{|I|/(q-1)}$}
\end{algorithm}

By construction of the algorithm, the fingerprints and containers satisfy 
\begin{align*}
    &F_0 \subseteq F_1 \subseteq F_2 \dots \subseteq F_{|I|/(q-1)} \subseteq I, \\
    I \subseteq \, C_t(F_t) &\cup F_t \text{ and } 
    C_{|I|/(q-1)} \subseteq C_{|I|/(q-1) -1} \dots \subseteq C_0 = V  
\end{align*}

We now prove some important properties of the containers that we use later in our proof to compute an upper bound on the sample complexity of $\rho$-independent set testing.
\begin{proposition}[Uniqueness of Containers]
    For any hypergraph $H = (V, E)$ with a predefined ordering of vertices $\pi$, any independent set $I$ in $H$, and any $t$, the fingerprint $F_t(I)$ and container $C_t(I)$ of $I$ satisfy
    \[  C_t \big( F_t(I) \big) = C_t(I). \]
\end{proposition}

\begin{proof}
    For $t \leq \frac{|I|}{q-1}$, the algorithm chooses the same fingerprint vertices in the same order for $F_t(I)$ and $I$ in the first $t-1$ rounds of the algorithm, breaking ties based on the given ordering of vertices $\pi$, hence the vertices removed from the container are also identical, i.e., $C_t \big( F_t(I) \big) = C_t(I)$. For $t > \frac{|I|}{q-1}$, $F_t(I) = I$, hence the inequality is satisfied. 
\end{proof}

We prove another important property of any $\ell$-uniform hypergraph which we refer as the standard averaging bound \cite{blais2024new}.

\begin{proposition}[Standard Averaging Bound]
\label{prop1}
    Let $H = (V, E)$ be an $\ell$-uniform hypergraph with $\ell \ge 2$. Then there are at least $\frac{\ell |E| - d |V|}{\binom{|V|-1}{\ell-1} - d} > \frac{\ell |E| - d \dot |V|}{\binom{|V|-1}{\ell-1} }$ vertices in $H$ whose degree is larger than $d$.
\end{proposition}

\begin{proof}
    Let S be the set of vertices with degree larger than d.
    
    The vertices in $S$ have degree at most $\binom{|V|-1}{\ell-1}$. The remaining $|V| - |S|$ vertices have degree at most $d$. Since the total number of edges is $1/\ell$ times the sum of degrees, the total number of edges $|E|$ is less than or equal to
    \[ \frac{1}{\ell} \left( |S| \binom{|V|-1}{\ell-1} + (|V| - |S|) d \right). \]
    Rearranging the terms we see that
    \[
        \left( |S| \left[ \binom{|V|-1}{\ell-1} - d \right] \right) \geq \ell |E| - d |V|
    \]
    which gives us the desired inequality.
\end{proof}

We can now find an upper bound on the maximum $q$-degree in any $\rho n$-induced subhypergraph of $H[C_t]$.

\begin{proposition}[Maximum Degree Bound]
\label{prop2}
    Let $H = (V, E)$ be a $q$-uniform hypergraph with $|V| = n$, let $I$ be an independent set in $H$, and let $t \le \frac{|I|}{q - 1}$. Then the maximum $q$-degree of every $(\le \rho n)$-subhypergraph of $H[C_t]$ is at most \[ \frac{2 (q-1)}{\rho t} \binom{\rho n - 1}{q - 1}. \]
\end{proposition}

\begin{proof}
    Fix any $t \le \frac{|I|}{q - 1}$. Since the containers satisfy
    $ V = C_0 \supseteq C_1 \supseteq \cdots \supseteq C_t, $
    we have
    \begin{align*}
        n \geq |C_0| - |C_t|
           = \sum_{i=1}^t \bigl( |C_{i-1}| - |C_i| \bigr)
           = \sum_{i=1}^t |C_{i-1} \setminus C_i|.
    \end{align*}
    
    Hence, there exists an index $j \in [t]$ such that $|C_{j-1} \setminus C_j| \le \frac{n}{t}$.
    
    \medskip
    \begin{claim}
        For each step of the inner loop, $1 \le \ell \le q-1$ in the $j^{th}$ iteration, 
        the number of $\ell$-edges in $\mathcal{D}_{j,\ell}$ is at most
        \begin{equation}\label{eqn1} 
            \frac{2\ell}{\rho t} \binom{\rho n - (q - \ell)}{\ell}
        \end{equation}
    \end{claim}
    \medskip
    
    \noindent\textit{Proof of Claim:} We argue by induction on $\ell$.
    
    \smallskip
    \emph{Base case: ($\ell_0$ = smallest value of $\ell$ such that $E[\mathcal{D}_{j,\ell+1}] > 0$} ) \\
    We observe that $\ell_0 = 1$ or $E[\mathcal{D}_{j,\ell_0}] = 0$ in this case.\footnote{We require $E[\mathcal{D}_{j,\ell_0}] = 0$ as a base case to handle the case that a vertex outside of $\mathcal{D}_{j,\ell+1}$ is added to the fingerprint, because when that happens $|V[\mathcal{D}_{j,\ell}]|$ does not shrink by one, and so the proof in the inductive step does apply as is. As such a vertex is selected then $E[\mathcal{D}_{j,\ell-1}]$ becomes $0.$}
    \begin{enumerate}[label=\roman*]
        \item  \textit{$\ell_0 = 1$:} The vertices corresponding to $1$-edges in $\mathcal{D}_{j,1}$ are removed when forming $C_j$, so their number is at most $|C_{j-1} \setminus C_j| \le \frac{n}{t}$. Since $q = o(\rho n)$ \footnote{We observe that the sample complexity bound of \Cref{thm1} holds trivially when $q = \Omega(\rho n)$. Notice that the number of vertices is $O\left(\frac{q}{\rho}\right)$ in this case so the sample complexity is $O\left(\frac{q}{\rho}\right) < O\left(\frac{1}{\varepsilon}\right)$. We remind the reader that the problem is non-trivial only when $\varepsilon < \frac{\rho^q}{q!}$. We can therefore assume $q = o(\rho n)$ for the remainder of this paper.}, we have \[\frac{n}{t} \le \frac{2}{\rho t} (\rho n - q + 1),\] so the claim holds for $\ell_0 = 1$.

        \item \textit{$E[\mathcal{D}_{j,\ell_0}] = 0$:} The claim trivially holds since $E[\mathcal{D}_{j,\ell_0}] = 0 < \frac{2\ell}{\rho t} \binom{\rho n - (q - \ell)}{\ell}$ for any value of $\ell.$
    \end{enumerate}
    \smallskip
    \emph{Inductive step:}  Assume \Cref{eqn1} holds for $\ell-1$ where $1\le \ell\le q-2$. Suppose, for contradiction,
    \[ |E[\mathcal{D}_{j,\ell}]| > \frac{2\ell}{\rho t}\binom{\rho n-(q-\ell)}{\ell} \]
    By induction, we have $|E[\mathcal{D}_{j,\ell-1}]| \leq \frac{2 (\ell-1)}{\rho t}\binom{\rho n - \left(q- (\ell - 1)\right)}{\ell - 1}$, \, hence \,
    \[ \deg_{\mathcal{D}_{j,\ell-1}}(v_{j,\ell-1}) \leq \frac{2(\ell-1)}{\rho t}\binom{\rho n-(q-\ell)-1}{\ell-1}. \]
    
    Therefore, by the standard averaging bound (\Cref{prop1}) with $d = (\ell - 1) \frac{|E|}{|V|}$, 
    \begin{align*}
        \frac{|E[\mathcal{D}_{j,\ell}]|}{\binom{|V[\mathcal{D}_{j,\ell}]|-1}{\ell-1}}
            > \frac{\frac{2\ell}{\rho t}\binom{\rho n-(q-\ell)}{\ell}}{\binom{\rho n - (q-\ell) -1}{\ell-1}}
            = \frac{2}{\rho t}\bigl(\rho n-(q-\ell)\bigr)
            > \frac{n}{t}
    \end{align*}
    where we have used the fact that $|V[\mathcal{D}_{j,\ell}]| \leq \rho n - (q-\ell)$.
    
    \medskip
    Hence there are more than $n/t$ vertices in some $\rho n$-induced subhypergraph of $\mathcal{D}_{j,\ell}$ with degree exceeding
    \begin{align*}      
        (\ell-1)\,\frac{|E[\mathcal{D}_{j,\ell}]|}{|V[\mathcal{D}_{j,\ell}]|}
            > (\ell - 1) \cdot \frac{\frac{2\ell}{\rho t} \binom{\rho n-(q-\ell)}{\ell}}{\rho n - (q-\ell)}
            = \frac{2(\ell-1)}{\rho t}\binom{\rho n-(q-\ell)-1}{\ell-1}
    \end{align*}
    
    All such vertices are removed from $C_j$ (since they have higher degree than $v_{j,\ell}$), implying \\
    $|C_{j-1} \setminus C_j|>n/t$, a contradiction. This completes the induction and the proof of the claim.
    
    \medskip
    \emph{Why the claim implies the proposition:}
    We consider the $j^{th}$ iteration of the algorithm. Taking $\ell = q-1$, the number of $(q-1)$-edges in $\mathcal{D}_{j,q-1}$ is at most
    $\frac{2\dot(q-1)}{\rho t} \binom{\rho n - 1}{q - 1}$. Therefore, the fingerprint vertex $v_{j,q-1}$ or $v_j$ has $q$-degree at most $\frac{2\dot(q-1)}{\rho t} \binom{\rho n - 1}{q - 1}$ in any $(\le \rho n)$-induced subgraph of $H[C_{j-1}]$. All vertices exceeding this degree in some $(\le \rho n)$-subgraph of $H[C_{j-1}]$ are excluded from $C_j$, so the maximum $q$-degree in any $(\le \rho n)$-subgraph of $H[C_j]$ is at most $\frac{2 (q-1)}{\rho t} \binom{\rho n - 1}{q - 1}$. Since $C_t \subseteq C_j$, the same bound holds for every $(\le \rho n)$-subgraph of $H[C_t]$.
\end{proof}

\subsection{Proof of the New Container Lemma}

We can now prove the new container lemma which essentially states that the container size is "small" (as defined in the lemma) for some small value of $t$ independent of the hypergraph size $n$ when the input hypergraph is $\varepsilon$-far from $\rho n$ independent set. We prove the contrapositive, i.e., we construct a $\rho n$-induced subhypergraph which is not $\varepsilon$-far from $\rho$-independent set property if the container size is "large" for all values of $t$.

\begin{lemma}[New Container Lemma]
    \label{lem1}
    Let $H = (V, E)$ be a $q$-uniform hypergraph on $n$ vertices that is $\varepsilon$-far from $\rho$-independent set property. Then, for any independent set $I$ in $H$, there exists an index $t \le \frac{8 \rho^{q-1}}{\varepsilon (q-1)!}$ such that the size of the $t^{th}$ container of $I$ satisfies
    \[ |C_t| \le \left( \rho - t \cdot \frac{\varepsilon (q-2)!}{4 \rho^{q-2} \ln(\rho^{q-1} / \left(\varepsilon (q-1)!\right))} \right) n \]
\end{lemma}

\begin{proof}
    We prove the contrapositive. Let $t_1 = \frac{8 \rho^{q-1}}{\varepsilon (q-1)!}$. Suppose that for all $t \le t_1$, the size of the $t^{th}$ container $C_t(I)$ satisfies
    \begin{equation} \label{eqn2}
        |C_t(I)| > \left(  \rho - t \cdot \frac{\varepsilon (q-2)!}{4 \rho^{q-2} \ln(\rho^{q-1} / (\varepsilon (q-1)!))}  \right) n
    \end{equation}
    We will show that there is some $\rho n$ induced subhypergraph with $< \varepsilon n^q$ edges showing that H is not $\varepsilon$-far from $\rho$-independent set property.
    
    \medskip
    
    We start with the container $C_{t_1}(I)$ to construct the set of vertices S and the corresponding induced subhypergraph $H[S]$. Add vertices one at a time from the previous containers $C_t(I)$ (in decreasing order of t) until the size of the set reaches $\rho n$. Note that we can choose any set of $\rho n$ vertices from $C_{t_1}(I)$ if the size of the container is greater than $\rho n$. Let $V_t$ denote the set of vertices in container $(t-1)$ but not in container $t$ in this set $S$, i.e., $V_t = \{v : v \in S \text{ and } v\in C_{t-1} \backslash C_t\}$ and let $t_0$ be the smallest value of t for which $|V_t| > 0$. 
    
    \medskip
    
    We now prove the claim. Since $\left| S \right| \le \rho n$, \Cref{prop2} implies that the maximum q-degree in any $\rho n$ induced subhypergraph of $H\left[C_{t_1}(I)\right]$ is at most
    \[  \frac{\frac{2(q-1)}{\rho} \binom{\rho n-1}{q-1}}{\frac{8 \rho^{q-1}}{\varepsilon (q-1)!}}
            \leq \frac{\varepsilon q!}{4 \rho^q} \binom{\rho n-1}{q-1} \]
    Hence, any $(\le \rho n)$-induced subhypergraph of $H\left[C_{t_1}(I)\right]$ has at most $\frac{\varepsilon q!}{4 \rho^q} \binom{\rho n}{q}$ $q$-edges.

    When $|C_{t_1}| \geq \rho n$, then the total number of edges is at most 
    \[ \frac{\varepsilon q!}{4 \rho^q} \binom{\rho n}{q} < \frac{\varepsilon q!}{\rho^q} \frac{(\rho n)^q}{q!} < \varepsilon n^q,\]
    therefore $H[S]$ is not $\varepsilon$-far from $\rho n$ independent set.

    Now we consider the case when $|C_{t_1}| < \rho n$.
    By \Cref{prop2}, each $v \in V_t$ contributes at most
    \[ \frac{2(q-1)}{\rho (t-1)} \binom{\rho n-1}{q-1} \]
    additional edges (or $\binom{\rho n-1}{q-1}$ if $t=1$), since $v$ lies in $C_{t-1}(I)$.
    
    \medskip
    
    Therefore, the total number of edges is at most
    \begin{equation}
        \label{eqn3}
        \frac{\varepsilon q!}{4 \rho^q} \binom{\rho n}{q}
        + |V_1| \cdot \binom{\rho n-1}{q-1}
        + \sum_{t=2}^{t_1} |V_t| \cdot \frac{2(q-1)}{\rho (t-1)} \binom{\rho n-1}{q-1}
    \end{equation}
    
    For any $t_0 \le t \le t_1$, the set $\bigcup_{j=1}^t V_j$ corresponds to the vertices not appearing in $C_t(I)$, and so \Cref{eqn2} implies
    \[ \sum_{j=1}^t |V_j| \le \rho n - |C_t| <  \frac{t \varepsilon n (q-2)!}{4\rho^{q-2} \ln(\rho^{q-1} / (\varepsilon (q-1)!))} \]
    
    In the sum~\eqref{eqn3}, the contribution from $|V_t|$ decreases as $t$ increases, so the sum is
    maximized when all $|V_t| = \frac{\varepsilon n (q-2)!}{4 \rho^{q-2} \ln(\rho^{q-1} / \varepsilon (q-1)!)}$. Substituting into~\eqref{eqn3}, we obtain
    \begin{align*}
    \frac{\varepsilon q!}{4 \rho^q} \binom{\rho n}{q}
    &+ \frac{\varepsilon n (q-2)!}{4 \rho^{q-2} \ln(\rho^{q-1} / (\varepsilon (q-1)!))} \binom{\rho n - 1}{q - 1} \\
    &+ \frac{\varepsilon n (q-2)!}{4 \rho^{q-2} \ln(\rho^{q-1} / (\varepsilon (q-1)!))} 
    \sum_{t=2}^{\frac{8 \rho^{q-1}}{\varepsilon (q-1)!}} \frac{2(q-1)}{\rho (t-1)} \binom{\rho n - 1}{q - 1} \\
    < \frac{\varepsilon q!}{4 \rho^q} \binom{\rho n}{q}
    &+ \frac{\varepsilon q!}{4 (q-1)\rho^{q-1} \ln(\rho^{q-1} / (\varepsilon (q-1)!))} \binom{\rho n}{q} \\
    &+ \frac{2 \varepsilon q!}{4 \rho^{q} \ln(\rho^{q-1} / (\varepsilon (q-1)!))} 
    \sum_{t=2}^{\frac{8 \rho^{q-1}}{\varepsilon (q-1)!}} \frac{1}{(t-1)} \binom{\rho n}{q} \\
    < \frac{\varepsilon q!}{\rho^q} \binom{\rho n}{q}  &\left( \frac{1}{4}
    + \frac{\rho}{4  (q-1) \ln(\rho^{q-1} / (\varepsilon (q-1)!))}
    + \frac{1}{2 \ln(\rho^{q-1} / (\varepsilon (q-1)!))}
    \sum_{t=2}^{\frac{8 \rho^{q-1}}{\varepsilon (q-1)!}} \frac{1}{t-1}  \right) \\
    < \frac{\varepsilon q!}{\rho^q} \binom{\rho n}{q}& < \frac{\varepsilon}{\rho^q} (\rho n)^q < \varepsilon n^q,
    \end{align*}
    where the final inequality uses the upper bound $H(m) \le \ln(m) + 1$ on the harmonic series and
    the assumption that $\varepsilon$ is less than a sufficiently large constant, i.e., $\varepsilon < e^{-8} \frac{\rho^{q-1}}{(q-1)!}$ (which can be made without loss of generality as we can always choose a smaller value of $\varepsilon'$ for the tester such that $e^{-8} \varepsilon < \varepsilon' < \varepsilon$ which gives us the same sample complexity). Note that $\varepsilon < \frac{\rho^q}{q!} < \frac{\rho^{q-1}}{(q-1)!}$ since the graph is $\varepsilon$-far from $\rho$-independent set property.
    
    This completes the proof.
\end{proof}

\begin{remark}
\label{remark:containerLemmaTightness}
    There exists a $q$-uniform hypergraph $H$ that is $\varepsilon$-far from having a $\rho n$ independent set and an independent set $I$ for which \Cref{lem1} is nearly tight in the $\rho$-dependence. Consider the same example from the lower bound of Langberg \cite{Langberg04}, and re-analyzed in \Cref{appendix1}, which is a hypergraph $H$ formed by planting an independent set $I$ of size $\left(\rho-\frac{4 (q-1)! \varepsilon}{\rho^{q-1}}\right) n$ in a complete $q$-uniform hypergraph. See \Cref{appendix1} for a justification that this hypergraph is $\varepsilon$-far from having a $\rho n$ independent set.
       
    We claim that for all $t$ we have that $|C_t(I)| \geq \left(\rho-\frac{16 (q-1)! \varepsilon t}{\rho^{q-2}}\right) n,$ which demonstrates that \Cref{lem1} is roughly tight in the $\rho$ dependence. To show this, first consider the case that $t \leq \frac{1}{2\rho}.$ Observe that for the first $\frac{1}{2\rho}$ iterations of the outer loop of the container procedure, exactly $\rho n$ vertices are removed from the container in each iteration. This is because no vertices will be removed due to having larger degree in a $\rho n$ subhypergraph (every vertex has full degree in some $\rho n$ subhypergraph). Hence, $|C_t(I)| \geq \rho n$ for all $t \leq \frac{1}{2\rho}.$
       
    Next consider the case that $t > \frac{1}{2\rho}.$ Recall that at every step $t$ of the container procedure we have that $I \subseteq C_t(I) \cup F_t(I),$ and so $|C_t(I)| \geq |I \setminus F_t(I)| \geq \left(\rho-\frac{8 (q-1)! \varepsilon}{\rho^{q-1}}\right) n,$ where we are assuming that $|F_t(I)|$ is sufficiently small. Hence, for $t > \frac{1}{2\rho},$ we find that $|C_t(I)| \geq \left(\rho-\frac{16 (q-1)! \varepsilon t}{\rho^{q-2}}\right) n,$ completing the proof of the claim.
\end{remark}

\begin{remark}
\label{remark:oldApproachWeakness}
    It is not possible to achieve our bound in \Cref{lem1} and \Cref{thm1} using a standard hypergraph container procedure.
    In particular, a standard hypergraph container procedure run on a $q$-uniform hypergraph $H=(V,E)$ and independent set $I$ would initialize a container to $V$ and run in iterations.
    In each iteration it would run an inner loop that selects a vertex $v$ in the independent set with highest degree in the hypergraph on the current container, then look at the $(q-1)$-uniform hypergraph formed by all edges incident to $v,$ and so on until reaching the $1$-uniform hypergraph. 
    This is in stark contrast to the algorithm described in \Cref{sect:ContainerMethod} that selects vertices with highest degree in some $\rho n$ subhypergraph.
    We argue that this natural hypergraph procedure is not useful for proving our bounds via an example described below.
    
    Similar to the the lower bound of Langberg \cite{Langberg04} and re-analyzed in \Cref{appendix1}, start by planting a large independent set $I$ of size $\left(\rho-\frac{8 (q-1)! \varepsilon}{\rho^{q-1}}\right) n$ in an n-vertex hypergraph.
    Between $I$ and $V \setminus I$ add all possible edges with $q-1$ endpoints in $I$ and $1$ endpoint in $V \setminus I.$
    The hypergraph $H[V \setminus I]$ is going to be a random hypergraph with edge density $\rho^{q-2} \cdot (q-1)$ so that the vertices in $V \setminus I$ have degree less than roughly $\rho^{q-2} (q-1) \cdot \binom{n}{q} \cdot \frac{q}{n} \leq \frac{\rho^{q-2} n^{q-1}}{(q-2)!},$ which is roughly the degree of vertices in $I.$
    Finally, select a small subset $I' \subset I$ of size roughly $q/\rho^{q-2}$ and remove all edges in the hypergraph with at least one endpoint in $I'.$
    Instead, for each $v \in I'$ select $2 \frac{\rho^{q-2}}{(q-2)!}n^{q-1}$ random $(q-1)$ tuples from $V \setminus v$ and add an edge containing $v$ and those $q-1$ vertices. Crucially we want the vertices in $I'$ to have highest degree in the hypergraph, but whose degree is similar to the vertices in $I.$
    
    It is not hard to show that $H$ is $\varepsilon$-far from having a $\rho n$ independent set with high probability (over the randomness of the edges added in the hypergraph), assuming that $\varepsilon$ is sufficiently small compared to $\rho.$
    Further, we claim that running the standard hypergraph container procedure above would take at least roughly $\frac{1}{\rho^{q-2}}$ steps to get a desirable container.
    This is because the vertices in $I'$ have highest degree and so they are selected in each step of the procedure until they have all been selected.
    Since the edges of vertices in $I'$ are randomly placed throughout the hypergraph with density roughly $2 (q-1) \rho^{q-2}$, the standard hypergraph container procedure removes roughly $2(q-1) \rho^{q-2} n$ vertices in each step of the outer loop.
    Hence it will take at least roughly $1/\rho^{q-2}$ steps until $|C_t(I)| < \rho n.$
    
    Comparing with our hypergraph procedure given in \Cref{sect:ContainerMethod} on this hypergraph construction, since at every step $t$ of the container procedure we have that $I \subseteq C_t(I) \cup F_t(I),$ then $|C_t(I)| \gtrsim \left(\rho-\frac{8 (q-1)! \varepsilon}{\rho^{q-1}}\right) n.$
    Hence, the $t$ given in \Cref{lem1} is at most $\frac{32 (q-1) \ln(1/\varepsilon)}{\rho},$ and so the number of steps required by a standard hypergraph is larger than the number of steps given in \Cref{lem1} by a factor of at least $1/\rho^{\Theta(q)}.$
\end{remark}

\section{Upper Bound on Sample Complexity}
\label{sect:thmProof}

The proof of the theorem uses the following form of Chernoff’s bound for hypergeometric distributions \cite{blais2025testing}.

\begin{lemma}[Chernoff's bound]
    \label{lem2}
    Let $X$ be drawn from the hypergeometric distribution $H(N, K, m)$, where $m$ elements are drawn without replacement from a population of $N$ elements, $K$ of which are marked. Let $X$ denote the number of marked elements drawn. Then for any $\vartheta \ge \mathbb{E}[X]$,
    \[  \Pr\bigl[X \ge \vartheta\bigr] \le  \exp\!\left(-\frac{(\vartheta - \mathbb{E}[X])^2}{\vartheta + \mathbb{E}[X]}\right). \] 
\end{lemma}

\upperBound*

\begin{proof}
    Let $S$ be a random set of $s = c \cdot \frac{q\rho^{2q-3}}{\varepsilon^2 (q-2)!^2} \ln^3\left(\frac{\rho^{q-1}}{\varepsilon (q-1)!}\right)$ vertices drawn uniformly at random from $V$ without replacement, where $c$ is a large enough constant. We ignore the integer rounding issues in the rest of the proof as they do not affect the asymptotics of the final result. The $\varepsilon$-tester checks if the induced subhypergraph $H[S]$ contains an independent size of $\rho s$, if yes, the tester accepts and rejects otherwise.

    If $G$ contains a $\rho n$ independent set, then $S$ contains at least $\rho s$ vertices from this independent set with probability at least $1/2$, since the number of such vertices follows a hypergeometric distribution, and the median of this distribution is at least $\rho s$.
    
    In the remainder of the proof we upper bound the probability that $S$ contains a $\rho s$-independent set when $G$ is $\varepsilon$-far from containing a $\rho n$-independent set. For the rest of the argument, let us call a container $C_t$ small when its cardinality is bounded by $|C_t| \le \left( \rho - t \cdot \frac{\varepsilon (q-2)!}{4 \rho^{q-2} \ln(\rho^{q-1}/(\varepsilon (q-1)!))} \right)n$, the expression in the lemma above.
    
    Let us denote the vertices of the sampled set $S$ as $u_1, u_2, \ldots, u_s$. First observe that, by \Cref{lem1}, for every independent set $I$ of size $\rho s$ in $S$, there is a $t \le \frac{8 \rho^{q-1}}{\varepsilon (q-1)!}$ and a fingerprint $F_t \subseteq I$ of size $(q-1) t$ that defines a small container $C_t(F_t)$ such that $I \subseteq C_t(F_t) \cup F_t $. Hence, the probability that $S$ contains an independent set of size at least $\rho s$ is at most the probability that there exists some $t \le \frac{8 \rho^{q-1}}{\varepsilon (q-1)!}$ such that $S$ contains a set of $t' = (q-1)t$ vertices that form the fingerprint $F_t$ of a small container $C_t$, and $S$ contains at least $\rho s - t'$ other vertices from $C_t$. We now upper bound this probability.
    
    For any $t'$ distinct indices $i_1, i_2, \ldots, i_{t'} \in [s]$, consider the event where the vertices $u_{i_1}, \ldots, u_{i_{t'}}$ form the fingerprint $F_t$ of a small container $C_t$. Let $X$ denote the number of vertices among the other $s - t'$ sampled vertices that belong to $C_t$. By the bound on the size of small containers, the expected value of $X$ is
    \begin{align*}
        E[X] &\le \left( \rho - t \cdot \frac{\varepsilon (q-2)!}{4 \rho^{q-2} \ln(\rho^{q-1}/(\varepsilon (q-1)!))} \right) (s - t') \\
        &< \rho s - \frac{t \varepsilon s (q-2)!}{4  \rho^{q-2} \ln(\rho^{q-1}/(\varepsilon (q-1)!))} \\
        &\le \rho s - t' - \frac{t \varepsilon s (q-2)!}{8 \rho^{q-2} \ln(\rho^{q-1}/(\varepsilon (q-1)!))},
    \end{align*}
    where the last inequality uses the fact that $s = c \cdot \frac{q\rho^{2q-3}}{\varepsilon^2 (q-2)!^2} \ln^3\left(\frac{\rho^{q-1}}{\varepsilon (q-1)!}\right)$ for a large enough constant $c$ and that the problem is non-trivial only when $\varepsilon < \frac{\rho^q}{q!}$.
    
    By the Chernoff bound (\Cref{lem2}), the probability that we draw at least $\rho s - t'$ vertices from $C_t$ in the final $s - t'$ vertices drawn to form $S$ is
    \begin{align*}
        \Pr[X \ge \rho s - t'] &\le \exp \left( -\frac{(\rho s - t' - E[X])^2}{\rho s - t' + E[X]} \right)  \\
        &\le \exp \left( - \frac{1}{2 \rho s} \frac{t^2 \varepsilon^2 s^2 (q-2)!^2}{64 \rho^{2q-4} \ln^2(\rho^{q-1}/(\varepsilon (q-1)!))} \right) \\
        &\le \exp \left( -\frac{t^2 \varepsilon^2 (q-2)!^2 s}{128 \rho^{2q-3} \ln^2(\rho^{q-1}/(\varepsilon (q-1)!))} \right).
    \end{align*} 
    Therefore, by applying a union bound over all possible choices of $t'$ indices from $[s]$, the probability that any set of at most $\frac{8 \rho^{q-1}}{\varepsilon (q-2)!}$ vertices in $S$ form the fingerprint of a small container from which we sample at least $\rho s - t'$ vertices in $S$ is at most
    \begin{align*}
        & \sum_{t'} \binom{s}{t'} \exp \left( -\frac{t'^2 \varepsilon^2 s}{128 (q-1)^2 \rho^{2q-3} \ln^2(\rho^{q-1}/(\varepsilon (q-1)!))} \right) \\
        \le & ~\sum_{t} \exp \left( (q-1)t \ln s - \frac{t^2 \varepsilon^2 (q-2)!^2 s}{128 \rho^{2q-3} \ln^2(\rho^{q-1}/(\varepsilon (q-1)!))} \right) \\
        < &~\frac{1}{3}
    \end{align*}
    where the first inequality uses the upper bound $\binom{s}{t'} \le s^{t'}$, and the last inequality again uses the fact that $s = c \cdot \frac{q\rho^{2q-3}}{\varepsilon^2 (q-2)!^2} \ln^3\left(\frac{\rho^{q-1}}{\varepsilon (q-1)!}\right)$ for a large enough constant $c$. Hence, the probability that $S$ contains an independent set of size at least $\rho s$ is less than $1/3$.
\end{proof}



\printbibliography

\appendix

\section{Lower Bound on Sample Complexity}
\label{appendix1}

We prove a lower bound for the sample complexity of $\rho$-independent set testing using the construction by Langberg \cite{Langberg04}.

\begin{proposition}
    Let $n$ be a sufficiently large constant. Let $\rho > 0$ and $\varepsilon > 0$ satisfy the following conditions:
    \begin{itemize}
        \item $\varepsilon \ll \dfrac{\rho^{q}}{2q!}$,
        \item $\dfrac{\rho^{2q-1}}{\varepsilon^{2}} \ll n$,
        \item $q^{2} \ll \rho n$,
    \end{itemize}
     
    Then there exists a hypergraph $H$ on $n$ vertices such that $H$ is $\varepsilon$-far from $\rho$-independent set property, and with probability at least $e^{-1/2}$, a random subhypergraph $H$ of size
    \[ s = \frac{\rho^{2q-1}}{8 \left((q-1)!\right)^{2} \varepsilon^{2}} \]
    contains an independent set of size $\rho s$.
\end{proposition}

\begin{proof}
    Consider the $q$-uniform hypergraph $H = (V,E)$ defined as follows:
    \begin{enumerate}[i)]
        \item $|V| = n$,
        \item $V$ is partitioned into two disjoint sets $A$ and $V \setminus A$, where
            \[ |A| = \left(1 - \frac{4(q-1)!\varepsilon}{\rho^{q}}\right)\rho n, \]
        \item the edge set $E$ consists of all $q$-subsets of $V$ except those fully contained
            in $A$ (i.e., $A$ is an independent set).
    \end{enumerate}
    
    We first show that H is $\varepsilon$-far from $\rho$-independent set property, i.e., 
    every subset of vertices of size $\rho n$ induces a subhypergraph
    with at least $\varepsilon n^{q}$ edges.

    Consider any subset $S \subseteq V $ of size $\rho n$. We define $\delta := \frac{4(q-1)!\varepsilon}{\rho^{q}}$ so $|A| = (1-\delta)\rho n$.

    Thus,
    \begin{align*}       
        |E(S)| &\ge \binom{\rho n}{q} - \binom{\left(1 - \delta\right)\rho n}{q} \\
               &\ge \frac{(\rho n - q)^q}{q!} - \frac{((1 - \delta)\rho n)^q}{q!} \\
               &\ge \frac{(\rho n)^q}{q!} \left[ \left(1 - \frac{q}{\rho n}\right)^q - \left(1 - \delta\right)^q \right] \\
               &\ge \frac{(\rho n)^q}{q!} \left[ \left(1 - \frac{q^2}{\rho n}\right) - \left(1 - \frac{\delta q}{2}\right) \right] \tag{*} \label{eqn4} \\ 
               &\ge \frac{(\rho n)^q}{q!} \left[ \left(\frac{\delta q}{2} - \frac{q^2}{\rho n}\right) \right] \\
               &\ge \frac{(\rho n)^q}{q!} \frac{\delta q}{4} \\
    \end{align*}
    where \eqref{eqn4} follows from $(1-a)^x \geq 1 - ax$ and $(1-a)^x \leq e^{-ax} \leq (1-\frac{ax}{2})$ when $ax < 1$ using Taylor's expansions and the last inequality follows when $q < \frac{\delta \rho n}{4}$. 
    
    Substituting $\delta = \frac{4(q-1)!\varepsilon}{\rho^q}$ gives $|E(S)| \ge \varepsilon n^q$.

    \medskip
    Now we consider a sample of $s = \frac{\rho^{2q-1}}{8 \left((q-1)!\right)^{2} \varepsilon^{2}}$ vertices chosen randomly from the hypergraph $H$.
    We define $X := |S \cap A| = \sum_{v \in A} \mathbf{1}_{\{v \in S\}}$.
    Then $X$ follows a hypergeometric distribution with parameters $(n, |A|, s)$, and mean $\mu = \frac{|A|}{n} \cdot s = \left(1 - \frac{4 (q-1)! \varepsilon}{\rho^q} \right)\rho s.$

    Thus, using the Central Limit Theorem, we obtain that for sufficiently large $n$,
    \begin{align*}
        \Pr\!\left[|S \cap A| > \left(1 - \frac{4(q-1)!\varepsilon}{\rho^{q}}\right)\rho s + \sqrt{\rho s} = \rho s \right]
        \geq \Pr[N(0,1) > 3/2]
        \geq \frac{1}{20}
    \end{align*}

    Therefore, the chosen sample has an independent set of size $\rho s$ with some constant positive probability and the sample size $s = \frac{\rho^{2q-1}}{8 \left((q-1)!\right)^{2} \varepsilon^{2}}$ gives a lower bound for testing $\rho n$ independent set.
\end{proof}

\end{document}